\documentclass[journal,letterpaper]{IEEEtran}

\usepackage{graphicx,cite,epsfig,amssymb,amsmath,amsfonts,multicol,subfigure,mathtools,bm,mathrsfs,setspace}
\usepackage{multirow}
\usepackage{xcolor}
\newtheorem{theorem}{Theorem}
\newtheorem{lemma}{Lemma}
\newenvironment{proof}{{\it Proof:}}{\hfill $\blacksquare$\par}

\begin{document}
	\title{A Practical Consideration on Convex Mutual Information}
	\author{
		Mingxi Yin, 
		Bingli Jiao, \emph{Senior Member}, \emph{IEEE},	
		Dongsheng Zheng	
		and Yuli Yang, \emph{Senior Member}, \emph{IEEE}
		\thanks{M. Yin and B. Jiao ({\em corresponding author}) are with the Department of Electronics, Peking University, Beijing 100871, China (e-mail: yinmx@pku.edu.cn, jiaobl@pku.edu.cn).}
		\thanks{Y. Yang is with the School of Engineering, University of Lincoln, Lincoln, U.K. (e-mail: yyang@lincoln.ac.uk).}
	}
	
	\maketitle
	
	\begin{abstract}
	 In this paper, we focus on the convex mutual information, which was found at the lowest level split in multilevel coding schemes with communications over the additive white Gaussian noise (AWGN) channel. Theoretical analysis shows that communication achievable rates (ARs) do not necessarily below mutual information in the convex region.  In addition, simulation results are provided as an evidence.
		
	\end{abstract}
	
	
	\IEEEpeerreviewmaketitle

\section{Convex Mutual Information}
The study of mutual information (MI) refers to entropies of the channel input and the noise~\cite{Shannon1948}.  This research refers to those involving transmissions of the multilevel coding schemes over the memoryless additive white Gaussian noise (AWGN) channel, whereat the overall MI is separated with respect to split signals~\cite{MLC_1977,MLC_1999,MLC_2000,MLC_2002,MLC_2007}. We restrict ourselves to work on the convex MI which can be found in the previous works~\cite{MLC_1999,MLC_2000,MLC_2007}. 

In multilevel coding schemes, there are several approaches to split the bit-to-symbol mapping of one modulated symbol into different levels, each of which constructs a coded modulation problem based on their individual MIs. The overall achievable rate (AR) is obtained by the summation of ARs of all levels and is limited by the MI~\cite{MLC_1999}. With this constraint, In exchanging for splitting out a higher MI of one level, those of the other levels must be lower.  Consequently, the convex MI 
presents with the lowest level at ow signal-to-noise ratio (SNR).     

Since the convex MI suppress the lowest level for its AR to an insignificant contribution, it makes sense to circumvent this upper bound.  Hence, we utilize a repeated transmission method to improve the lowest level. Because that MI of every level is approximately in straight line at low SNR \cite{Verdu2011}, the negative effect of the repetition is small.
         
For showing the problem explicitly, we split the Quadrature Phase-Shift Keying (QPSK) constellation for providing an example of convex MI in the multilevel coding scheme.
The QPSK constellation is split into two levels in complex plan: a low level and a high level.  The two constellations are shown in Fig. \ref{fig1}, where the low level mapping is shown in Fig. \ref{fig1b} and the high level mapping is shown in Fig. \ref{fig1c}.  
As shown in Fig. \ref{fig1a}, the alphabet for the QPSK constellation can be denoted by ${\cal Q}{\rm{ = }}\{ {q_1},{q_2},{q_3},{q_4}\} $, where ${q_1} = A + j0$, ${q_2} = 0 + jA$, ${q_3} =  - A + j0$, ${q_4} = 0 - jA$. Denote the binary information sequence by the $1 \times 2$ vector ${\bf{v}} = [{v^H},{v^L}]$, ${v^H},{v^L} \in \{ 0,1\} $.

\begin{figure}
	\centering
	\subfigure[Actual channel]{
		\includegraphics[width=0.42\textwidth]{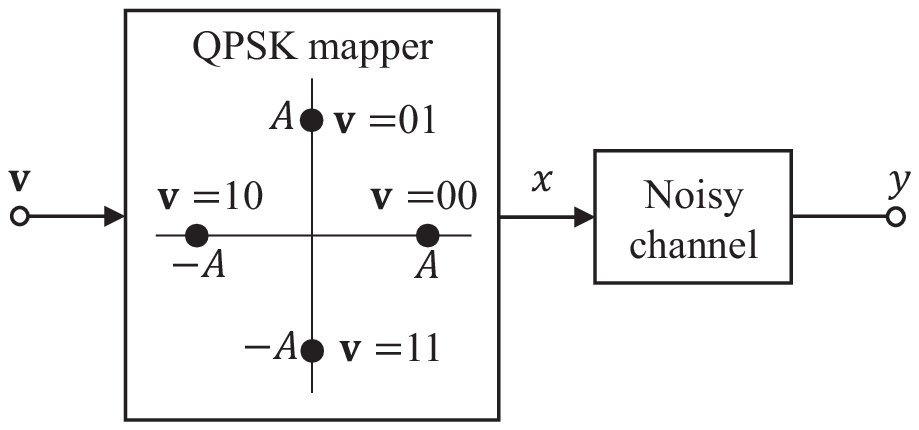}
		\label{fig1a}}
	\subfigure[Equivalent channel for the low level]{
		\includegraphics[width=0.42\textwidth]{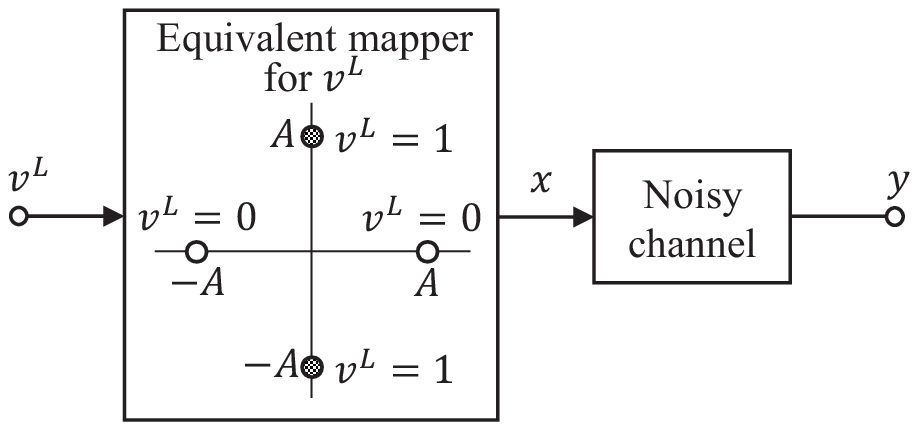}
		\label{fig1b}}
	\subfigure[Equivalent channel for the high level]{
		\includegraphics[width=0.42\textwidth]{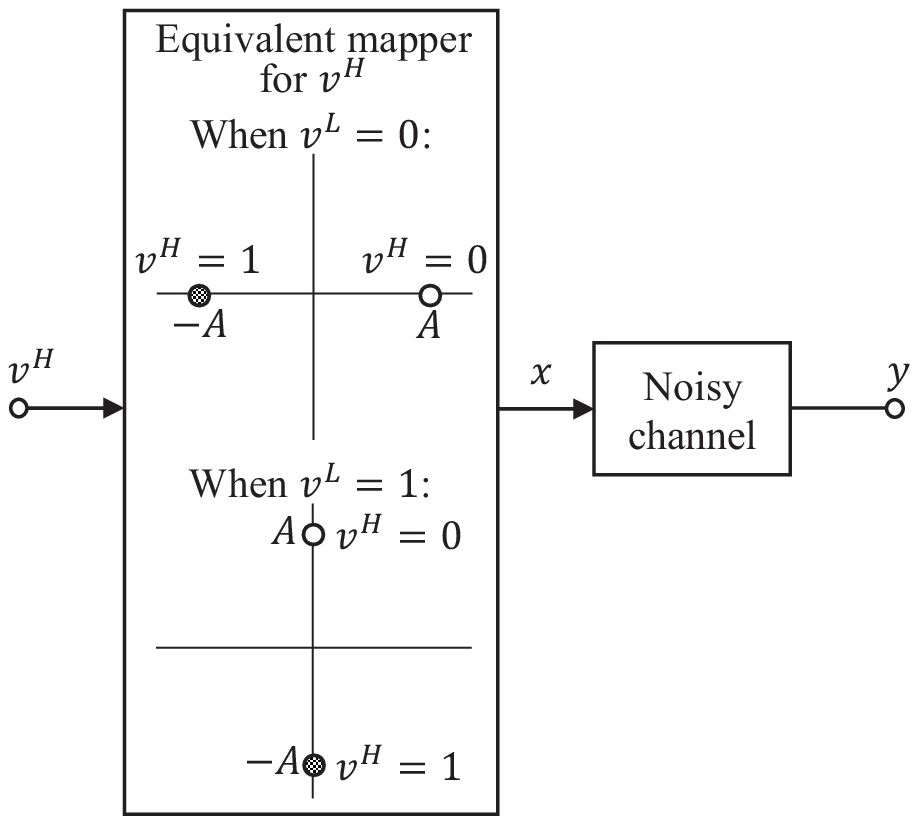}
		\label{fig1c}}
	\caption{Equivalent channels for multilevel coded QPSK modulation. }
	\label{fig1}
\end{figure}

The design of MLC for the QPSK modulation is explained as follows.
In the first step, at the low level, the signal set ${\cal Q}$ is divided into two parts, namely, the subsets ${\cal Q}({v^L} = 0) = \{ {q_1},{q_3}\} $ and ${\cal Q}({v^L} = 1) = \{ {q_2},{q_4}\} $. Each subset at the low level is uniquely labeled by the
path ${v^L}$. Then each of these two subset are divided into two further subsets ${\cal Q}({v^L},{v^H} = 0)$ and ${\cal Q}({v^L},{v^H} = 1)$ at the high level, and each subset at this level is uniquely labeled by the
path ${v^L}{v^H}$. At the high level of QPSK each subset only contains one signal point, concretely, we have ${\cal Q}({v^L} = 0,{v^H} = 0) = \{ {q_1}\} $,
${\cal Q}({v^L} = 0,{v^H} = 1) = \{ {q_3}\} $, ${\cal Q}({v^L} = 1,{v^H} = 0) = \{ {q_2}\} $,  ${\cal Q}({v^L} = 1,{v^H} = 1) = \{ {q_4}\} $.
Therefore, by $[{v^H},{v^L}]$ the transmit signal can by obtained by $x={\cal Q}({v^L},{v^H})$. Over the AWGN channel, the receive signal is given by 
\begin{equation}\label{1-1}
y = x + n 
\end{equation}
which can be modelled by the separated low level and high level channel as shown in Fig. \ref{fig1}. At the low level, $x \in \{-jA,jA\}$ for $v^L = 1$ and  $x \in \{-A,A\} $ for $v^L=0 $, respectively, where $j=\sqrt{-1}$, $x$ is the signal in Euclidean space and $v^L$ is the bit in Hamming space at the low level of the separation, $n \sim \mathcal{CN}(0,\sigma^2)$ denotes the AWGN and $y$ denotes the channel output.

MI of the low level, i.e., MI between the low level information bit $v^L$ and the receive signal $y$, is calculated by~\cite{Shannon1948}
\begin{eqnarray}\label{1-2}
	\begin{aligned}
		I(V^L,Y) =& \int_{ - \infty }^\infty    \sum\limits_x  \Bigg[ P_{V^L}(V^L)P_{Y|V^L}(y|v^L) \Biggr. \\
		& \qquad \qquad \times \Biggl. {\log _2}\frac{P_{V^LY}(v^L,y)}{P_{V^L}(v^L){P_{Y}}(y)} \Bigg] {dy} \\
		=& {\log _2}2 - \frac{1}{2}\int_{ - \infty }^\infty  \sum\limits_{d = 1}^2  \Bigg[ P_{Y|V^L_d}(y|v^L_d) \Biggr.\\ 
		& \qquad \qquad \times \Biggl. \log _2\frac{\sum\limits_{k = 1}^2 P_{Y|V^L_k}(y|{v^L_k}) }{P_{Y|V^L_d}(y|{v^L_d})} \Bigg]  {dy} 
	\end{aligned}
\end{eqnarray}
where $V^L$ and $Y$ are random variables for the information bit of low level, i.e., $v^L$, the channel output $y$, respectively, $P(\cdot)$ is density function of the probability, and $n$ is the AWGN with $n \sim \mathcal{N}(0,\sigma^2)$, respectively. The conditional probabilities $P_{Y|V^L_1}(y|{v^L_1})$ and $P_{Y|V^L_2}(y|{v^L_2})$ used in \eqref{1-2} are given by
\begin{equation}\label{pv1}
	\begin{aligned}
		&P_{Y|v^L_1}(y|{v^L_1}) = P_{Y|V^L=0}(y|{v^L=0})=\\ 
		&\frac{1}{{2\pi \sigma ^2}}\left( {\exp \left( { - \frac{{{{\left\| {{y} - {q_1}} \right\|}^2}}}{{\sigma ^2}}} \right) + \exp \left( { - \frac{{{{\left\| {{y} - {q_3}} \right\|}^2}}}{{\sigma ^2}}} \right)} \right)
	\end{aligned}
\end{equation}
and
\begin{equation}\label{pv2}
	\begin{aligned}
		&P_{Y|V^L_2}(y|{v^L_2}) = P_{Y|V^L=1}(y|{v^L=1})=\\
		&\frac{1}{{2\pi \sigma ^2}}\left( {\exp \left( { - \frac{{{{\left\| {{y} - {q_2}} \right\|}^2}}}{{\sigma ^2}}} \right) + \exp \left( { - \frac{{{{\left\| {{y} - {q_4}} \right\|}^2}}}{{\sigma ^2}}} \right)} \right)
	\end{aligned}
\end{equation}
respectively.

Then \eqref{1-2} can be derived as a function of the given SNR $\gamma$ as
\begin{equation}\label{eq_QMI_L}
	\begin{array}{l}
		\mathcal{I}^L(\gamma)= 1 \\
		-\frac{1}{4}{\mathbb{E}_W}\left[ {\log \left( {1 + \frac{{f\left( {W, - \sqrt \gamma   + j\sqrt \gamma  } \right) + f\left( {W, - \sqrt \gamma   - j\sqrt \gamma  } \right)}}{{f\left( {W,0} \right) + f\left( {W, - 2\sqrt \gamma  } \right)}}} \right)} \right]\\
		{\rm{                   }} - \frac{1}{4}{\mathbb{E}_W}\left[ {\log \left( {1 + \frac{{f\left( {W,\sqrt \gamma   + j\sqrt \gamma  } \right) + f\left( {W,\sqrt \gamma   - j\sqrt \gamma  } \right)}}{{f\left( {W,2\sqrt \gamma  } \right) + f\left( {W,0} \right)}}} \right)} \right]\\
		{\rm{                   }} - \frac{1}{4}{\mathbb{E}_W}\left[ {\log \left( {1 + \frac{{f\left( {W,\sqrt \gamma   - j\sqrt \gamma  } \right) + f\left( {W, - \sqrt \gamma   - j\sqrt \gamma  } \right)}}{{f\left( {W,0} \right) + f\left( {W, - j2\sqrt \gamma  } \right)}}} \right)} \right]\\
		{\rm{                   }} - \frac{1}{4}{\mathbb{E}_W}\left[ {\log \left( {1 + \frac{{f\left( {W,\sqrt \gamma   + j\sqrt \gamma  } \right) + f\left( {W, - \sqrt \gamma   + j\sqrt \gamma  } \right)}}{{f\left( {W,j2\sqrt \gamma  } \right) + f\left( {W,0} \right)}}} \right)} \right]
	\end{array}
\end{equation}
where ${\mathbb E}[\cdot]$ denotes the expectation operator, $W \sim \mathcal{CN}(0,1)$ denotes the random variable of AWGN, and the function $f(W,a)$ in \eqref{eq_QMI_L} is given by
\begin{equation}\label{1-10}
	f(W,a) = e^{- {\left( W - a \right)}^2} 
\end{equation} .
	
By scaling SNR in the linear manner, numerical results of \eqref{eq_QMI_L} are plotted as a function of linear SNR in Fig. 2, where the convexity of MI for the low level can be found in the SNR region of $\gamma$ in $[0,1.5]$. 

\begin{figure}
	\centering
	\includegraphics[width=0.5\textwidth]{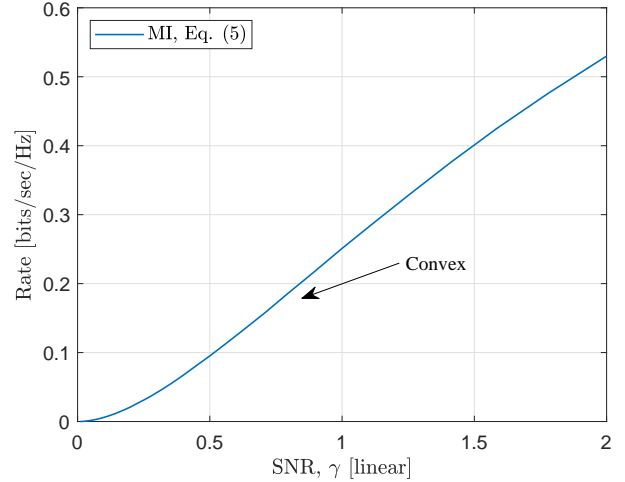}
	\caption{MI for the low level. }
	\label{fig_MI}
\end{figure}

To work on the convex problem mathematically, the definition of convex function is recalled as follows.

For $x_b>x_a$, if 
\begin{equation}\label{1-3}
f(x) < y=\frac{y_b-y_a}{x_b-x_a}x
\end{equation}
holds, $f(x)$ is convex with $x \in [x_a, x_b]$, where $x_b$ and $x_a$ are two arguments at horizontal axis and $y$ a straight line.

We then prove that the AR of low level is higher than the MI of low level in the convex region of MI.    

\section{Theoretical Analysis and Simulation Confirmation}
In this section, the theoretical work uses the definition of AR in the concept of the transmit bit rate at ``arbitrary small'' error probability, and the simulation uses bit error rate (BER) of $10^{-6}$ as the approximation of the ``arbitrary small'' with the AR. 

The theoretical proof is given in the following subsection and the simulation results are presented next.

\subsection{Theoretical Proof}
\begin{lemma}
For a given signal modulation, the error probability of the transmission keeps unchanged when 
\begin{equation}\label{2-1-1}
\mathcal{R}(\gamma) = \kappa \gamma
\end{equation}
for $\gamma= \gamma_1/M$, where $M$ is a positive integer, $\mathcal{R}(\gamma)$, $\kappa$ and $\gamma$ are the transmission rate of information bits, a constant and the SNR, and $\gamma_1$ is the SNR at $M=1$, respectively.     
\end{lemma}

\begin{proof}
Let us work in AWGN channel model   
\begin{equation}\label{2-1-2}
y= x+n
\end{equation}
with $y$, $x$ and $n$ are the channel output, the channel input and the AWGN component, respectively, where $n\sim \mathcal{N}(0,{\sigma ^2})$.

Let $ x'= x/\sqrt{M}$, repeating $x'$ for $M$ times and input all repeated $x'$ into \eqref{2-1-2} changes \eqref{2-1-2} in the vector form   
\begin{equation}\label{2-1-3}
\begin{aligned}
[y'_1,y'_2,&\cdots,y'_m,\cdots,y'_M] =\\
& \underbrace {[x', \cdots ,x',{\rm{ }}x']}_M + [n'_1,n'_2,\cdots,n'_m,\cdots,n'_M]
\end{aligned}
\end{equation}
for $m=1,2,\cdots,M$, where $y'_m$ and $n_m$ are the $m^{th}$ components of channel output and the associate Gaussian noise, respectively.  We note that $n'_m$ is statistically the same noise as that in \eqref{2-1-2}.  

At the channel output, the receiver sums over all components and obtain the demodulation equation as 
\begin{equation}\label{2-1-4}
y' = Mx' + \sum_m{n'_m} 
\end{equation}
where $y'$ is the result of the summation and $Mx'$  is the demodulated signal.  

Since the SNRs in \eqref{2-1-4} and \eqref{2-1-2} are same, the error probability of the former is exactly same as that of the latter with, however, its rate reduction of factor $1/M$.  Thus, one can find that \eqref{2-1-1} holds in general.   
\end{proof}

\begin{theorem} 
ARs of the repeated low level can be located in a straight line geometrically drawn from zero to a point of MI. Thus, there must exist a AR larger than the MI in the convex region defined by \eqref{1-3}.  
\begin{equation}\label{2-1}
\mathcal{R}^L(\gamma) = \frac{\mathcal{R}^L(\gamma_1)}{\gamma_1} \gamma 
\end{equation} 
for $\gamma = \gamma_1/M$, where $M$ is a positive integer number, $\mathcal{R}^L(\gamma)$ is the function for the AR of low level at SNR $\gamma$.  

\end{theorem}

\begin{proof}
	According to Shannon theory, there exist the capacity achieving codes that allows the AR of the low level to approach the MI calculated in \eqref{eq_QMI_L} at a negligible gap, whereby we write the approximation as
	\begin{equation}\label{2-1-8}
	\mathcal{R}^L (\gamma_1) = I(V^L;Y)|_{\gamma=\gamma_1} = \mathcal{I}^L(\gamma_1)
    \end{equation}	
	where $I(V^L,Y)|_{\gamma=\gamma_1}$ denotes the MI of low level at $\gamma=\gamma_1$.
	
	Using Lemma 1 to $\mathcal{R}^L$ yields   
	\begin{equation}\label{2-8}
	\mathcal{R}^L(\gamma_M)=\frac{\mathcal{R}^L(\gamma_1)}{\gamma_1} \gamma_M =\frac{\mathcal{I}^L(\gamma_1)}{\gamma_1} \gamma_M 
	\end{equation}
	
\end{proof}
    
    The numerical results of \eqref{2-8} are plotted for $M = 1, 2, 4, 8 $ in Fig. \ref{fig_AR} to provide an intuitive view of
    \begin{equation}\label{eq_R_M}
    	\mathcal{R}^L(\gamma_M)>\mathcal{I}^L(\gamma_M)
    \end{equation} 
    when $\gamma_M$ in the convex region $[0,1.5]$, due to that Fig. \ref{fig_MI} shows that in this region 
    \begin{equation}\label{eq_I_M}
    	\mathcal{I}^L(\gamma_M) < \frac{\mathcal{I}^L(\gamma_1)}{\gamma_1} \gamma_M 
    \end{equation} 
	
    Finally, since \eqref{2-8} holds for the relationship between the MI of a signal modulation and the AR of the repeated modulate signal in general, there is an insignificant difference between the MI and AR when MI curve is geometrically close to a straight line.  It is lucky that the MI of higher level agrees with the straight line at low SNR~\cite[Theorem 1]{Verdu2011}. This issue would be in our future interests.

\subsection{Practical Simulation}
 	To confirm the theoretical proof in the above section, simulations are performed using MATLAB.

 	LDPC codes with the code length of 64800 are selected from DVB-S.2 standard to simulate the BER performance of the low level as given in \eqref{1-1}, which is obtained by splitting the QPSK constellation. In simulations, different code rates are adopted for searching simulated SNR denoted by $\hat \gamma_1$ for BER at $10^{-6}$.  
 	
 	In decoding procedures of the low level, soft decision based on log–likelihood ratios (LLRs) is used in the signal demodulation, given by
 	\begin{eqnarray}\label{eq_LLR} 
 	\begin{aligned}
 	&{\rm{LLR}}^L =\ln \frac{\displaystyle {\sum\limits_{v^L = 0} {\exp \left( { -   \frac{{{{\left\| {y -  x } \right\|}^2}}}{{\sigma^2}}} \right)} }}{\displaystyle {\sum\limits_{v^L = 1} {\exp \left( { -  \frac{{{{\left\| {y - x} \right\|}^2}}}{{\sigma^2}}} \right)} }} \\
 	&=\ln \frac{\displaystyle {{\exp \left( { -   \frac{{{{\left\| {y -  A } \right\|}^2}}}{{\sigma^2}}} \right)} + {\exp \left( { -   \frac{{{{\left\| {y + A } \right\|}^2}}}{{\sigma^2}}} \right)} }}{\displaystyle  {\exp \left( { -  \frac{{{{\left\| {y - jA} \right\|}^2}}}{{\sigma^2}}} \right)}+{\exp \left( { -  \frac{{{{\left\| {y + jA} \right\|}^2}}}{{\sigma^2}}} \right)} } .
 	\end{aligned}
 	\end{eqnarray}
 The LLR of the high level is given by
 \begin{equation}\label{4-2}
 	{{\rm{LLR}}^H} = \left\{ \begin{array}{l}
 		\ln \frac{{\exp \left( { - \frac{{{{\left\| {y - {A}} \right\|}^2}}}{{\sigma ^2}}} \right)}}{{\exp \left( { - \frac{{{{\left\| {y +A} \right\|}^2}}}{{\sigma ^2}}} \right)}},\ {{\rm{LLR}}^L} > 0\\
 		\ln \frac{{\exp \left( { - \frac{{{{\left\| {y - jA} \right\|}^2}}}{{\sigma ^2}}} \right)}}{{\exp \left( { - \frac{{{{\left\| {y +jA} \right\|}^2}}}{{\sigma ^2}}} \right)}},\ {{\rm{LLR}}^L} \le 0
 	\end{array} \right.
 \end{equation}
 	
 	In simulations, we first set a code rate $R_1=1/2$ for the low level and find $\hat \gamma_1 =2.10 = 3.23$dB at which the BER of the low level achieves $10^{-6}$ as shown in Fig. 4.  Then, we set $ R_M = 1/4, 1/8, 1/16$ (according to $M=2,4,8$ in \eqref{2-8}) and obtained BER results are shown in Fig. 4 as well.  
 	
 	By picking up all SNRs with respect to the BER at $10^{-6}$ for the cases that $M=1,2,4,8$, we plot these points and the responding ARs in Fig. 3, for the comparison with theoretical results of MI. In this figure, two points of simulated ARs, i.e., ARs for the cases that $M=4,8$, are beyond the curve of the convex MI.  SNR gains between simulated ARs and the convex MI are found at $0.57$dB when $M=4$, and $1.62$dB when $M=8$.

 	\begin{figure}
 		\centering
 		\includegraphics[width=0.5\textwidth]{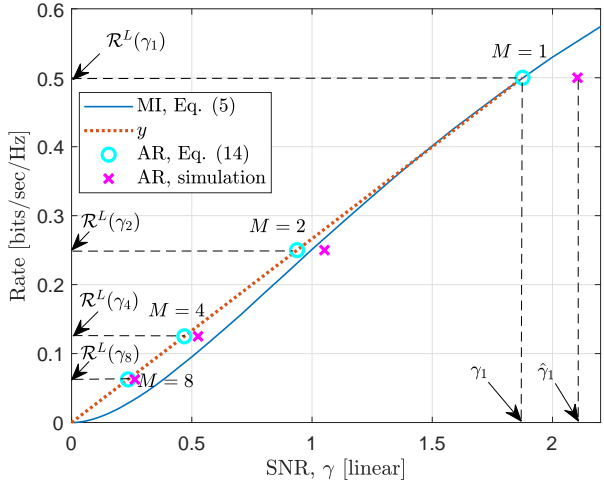}
 		\caption{Simulated ARs for the low level.  }
 		\label{fig_AR}
 	\end{figure}

 	\begin{figure}
 		\centering
 		\includegraphics[width=0.5\textwidth]{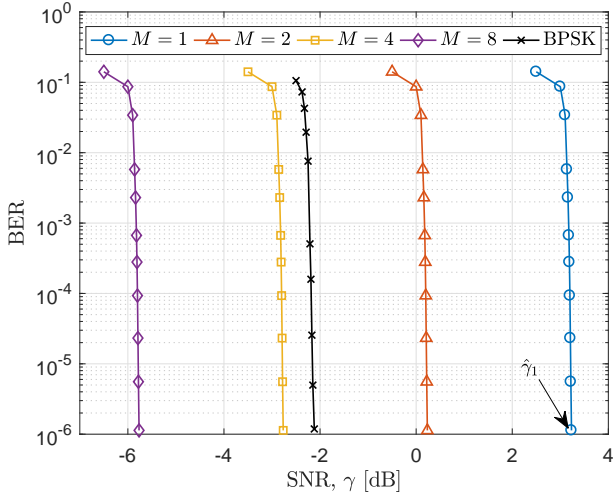}
 		\caption{BER performance for the low level. }
 		\label{fig_BER}
 	\end{figure}
 
\section{Conclusion}
The present work shows that ARs of the low level in multilevel coding are not necessarily limited by the convex MI with split signals. In the theoretical work, we have proved that ARs of repeated signals of the low level can be in a straight line when approaching zero SNR instead of along with the convex MI of the low level.  Simulation results have confirmed the straight line of the AR and its possibility beyond the MI.

\end{document}